\documentclass[11pt,oneside,english]{amsart}
\usepackage[T1]{fontenc}
\usepackage[latin9]{inputenc}
\usepackage{geometry}
\usepackage{amsthm}
\usepackage{amssymb}
\usepackage{graphicx}
\usepackage{esint}

\makeatletter
\numberwithin{equation}{section}
\numberwithin{figure}{section}
\theoremstyle{plain}
\newtheorem{thm}{\protect\theoremname}
  \theoremstyle{plain}
  \newtheorem{prop}[thm]{\protect\propositionname}
  \theoremstyle{remark}
  \newtheorem{rem}[thm]{\protect\remarkname}

\makeatother

  \providecommand{\propositionname}{Proposition}
  \providecommand{\remarkname}{Remark}
\providecommand{\theoremname}{Theorem}

\begin{document}

\title{Exact renormalization group analysis of turbulent transport by
the shear flow}

\author{Weinan E$^1$
and Hao Shen$^2$
\bigskip
\\
\SMALL{$^1 $School of Mathematical Sciences and BICMR,  Peking University}
\\
\SMALL{Department of Mathematics and 
PACM, Princeton University,}
\\
\SMALL{e-mail: weinan@math.princeton.edu}
\smallskip
\\
\SMALL{$^2 $Program in Applied and Computational Mathematics,
Princeton University}
\\
\SMALL{e-mail: hshen@princeton.edu}
\\
}

\begin{abstract}
The exact renormalization group (RG) method
initiated by Wilson and further developed by Polchinski is
used to study the shear flow model
proposed by Avellaneda and Majda as a simplified model for the diffusive transport of 
a passive scalar by a turbulent velocity field. 
It is shown that this exact RG method is capable of recovering all
the scaling regimes as the spectral parameters of velocity statistics
vary, found by Avellaneda and Majda in their rigorous study of this model.
This gives further confidence that the RG
method, if implemented in the right way instead of using drastic truncations
as in the Yakhot-Orszag's approximate RG scheme, does give the correct
prediction for the large scale behaviors of solutions of stochastic 
partial differential equations (PDE).
We also derive the analog of the ``large eddy simulation'' models when
a finite amount of small scales are eliminated from the problem.
\end{abstract}

\maketitle

\section{Introduction}

The renormalization group method is the most powerful tool for studying the
infra-red and ultra-violet behavior of complex systems.  It has completely 
revolutionized the way we study the critical behavior of models in statistical
mechanics as well as the continuum limit of models in quantum field theory.
However, the application of RG to stochastic partial differential equations 
has yielded limited success. The most important example in this direction is
that of hydrodynamic turbulence. The hope has been that by applying RG on 
the Navier-Stokes equation with noise, we can predict the key features of turbulent
flows such as their energy spectrum, structure functions, etc.
In this regard, the most notable attempt has been that of Yakhot and Orszag
\cite{Yakhot_Renormalization_1986}, who conducted such a study and obtained explicit predictions 
by adopting a simplified version of the RG scheme developed earlier by 
Forster, Nelson and Stephen for stochastic PDEs \cite{Forster_Large_1977}.
However,  the validity of their work has been questioned ever since, and
alternative RG approaches have been proposed, see for example \cite{Eyink_Renormalization_1994,smith1992yakhot}.
The most definitive results are found in the work of Avellaneda and Majda:
They proposed and analyzed a class of turbulent diffusion models and compared
their rigorous results with that of the results obtained
using the Yakhot-Orszag's RG scheme on the same  problem.
They found that the Yakhot-Orszag scheme recovers the correct results only
in the so-called mean field regime and failed in all the other regimes.
This has cast further doubt on  the usefulness of the RG approach outside of 
statistical mechanics and quantum field theory.

We will argue in this note that the problem identified by Avellaneda and Majda
is due to the failure of the Yakhot-Orszag's approximate RG scheme, not the
RG method itself.  By adopting the exact RG method, advocated originally by
Wilson and further developed for field theory by Polchinski, we show that the
rigorous results of Avellaneda and Majda can be recovered for all the regimes.
Our work gives much needed confidence for the usefulness of the RG method
outside of its traditional domains of statistical mechanics and quantum field theory.

We remark that our intention is to illustrate the application of the exact
RG method, not to give another rigorous proof of the results of Avellaneda and Majda.
Therefore we will focus our attention on the exact RG formalism, instead of making
all of our statements rigorous.
However, it should be noted that there are no essential difficulties in adding all
the mathematical details needed to make things rigorous.


\section{The Avellaneda-Majda model}

In \cite{Avellaneda_Mathematical_1990} Avellaneda and Majda
proposed the following model for the diffusive transport of a passive scalar
by a turbulent velocity field:
\begin{equation}\label{eq:AMmodel}
\frac{\partial T^{\delta}}{\partial t}+v_{\delta}(x,t)\frac{\partial T^{\delta}}{\partial y}=\frac{1}{2}\nu_{0}\Delta T^{\delta}
\qquad T^{\delta}\big|_{t=0}=T_{0}(\delta x,\delta y)
\end{equation}
We will explain the notations in this model in a minute. But let us remark
immediately that the main simplifying feature of this model is that the velocity
field is a shear flow.
In addition, we assume the velocity field $v$ is a Gaussian random field with mean zero and 
\begin{equation}
\left\langle \left|\hat{v}_{\delta}(k)\right|^{2}\right\rangle =\sqrt{2\pi}1_{\delta\leq|k|\leq1}\left|k\right|^{1-\epsilon}\quad\mbox{(steady case)}
\end{equation}
\begin{equation}
\left\langle |\hat{v}_{\delta}(k,\omega)|^{2}\right\rangle =\sqrt{2\pi}1_{\delta\leq|k|\leq1}|k|^{1-\epsilon}\frac{|k|^{z}}{\omega^{2}+|k|^{2z}}\quad\mbox{(unsteady case)}
\end{equation}
Here $\hat{v}$ is the Fourier transform of $v$, 
$1_{\delta\leq|k|\leq1}$
is an indicate function which serves as a cut-off function with the infrared cutoff 
being $\delta>0$ and ultraviolet cut-off being $1$, $\nu_{0}$ is the bare diffucivity 
constant,  
$T_{0}$ is a given (smooth) function which serves as the initial condition.
When there is no danger of confusion, we omit the subscripts $\delta$
on $T$ and $v$.

The key parameters in this model are $\epsilon$ and $z$ which characterize
the spectral properties of the velocity field and satisfy the constraints
$-\infty<\epsilon<4$, $z\geq0$ 
(when $\epsilon\geq 4$, the infrared divergence is too severe that 
there is no way whatsoever to obtain a large-scale limit).
$\epsilon$ controls how energy is distributed in the scales.
In particular, if $\epsilon \ge 2$, the kinetic energy density is infinite due to 
the concentration of energy at the large scales. It is this feature that gives rise
to interesting scaling properties for this model.
$z$ measures how fast the velocity field decorrelates in time.
The well-known Kolmogorov spectrum corresponds to the values $\epsilon = 8/3, z=2/3$.

Let
\begin{equation}
\bar{T}(x,y,t)=\lim_{\delta\rightarrow0}\left\langle T^{\delta}
\left(\frac{x}{\delta},\frac{y}{\delta},\frac{t}{\delta^{\alpha}}\right)
\right\rangle \label{eq:main_question}
\end{equation}
The main questions are:
\begin{enumerate}
 \item For what values of $\alpha$ this limit exists and is non-trivial?
 \item Identify the effective model that governs the limiting $\bar{T}$.
\end{enumerate}

Avellaneda and Majda identified  three
scaling regimes in $\epsilon$ for the steady case and five regimes in $(\epsilon,z)$
in the unsteady case. In a second paper \cite{Avellaneda_Approximate_1992}, Avellaneda and
Majda  applied the Yakhot-Orszag's 
approximate RG method developed in \cite{Forster_Large_1977,Orszag_Analysis_1999,Yakhot_Renormalization_1986}
to study the same model for the unsteady case.
Contrast to their rigorous results,
they were only able to find three regimes in $(\epsilon,z)$.
Except for the mean field regime, the results of the approximate RG method
do not match that of the exact results.
The accompanying figure, essentially taken from \cite{Avellaneda_Approximate_1992},
summarizes the situation. The left figure is the phase diagram from the exact
results.  The right figure is the phase diagram predicted using the approximate RG.


\begin{center}\includegraphics[scale=0.5]{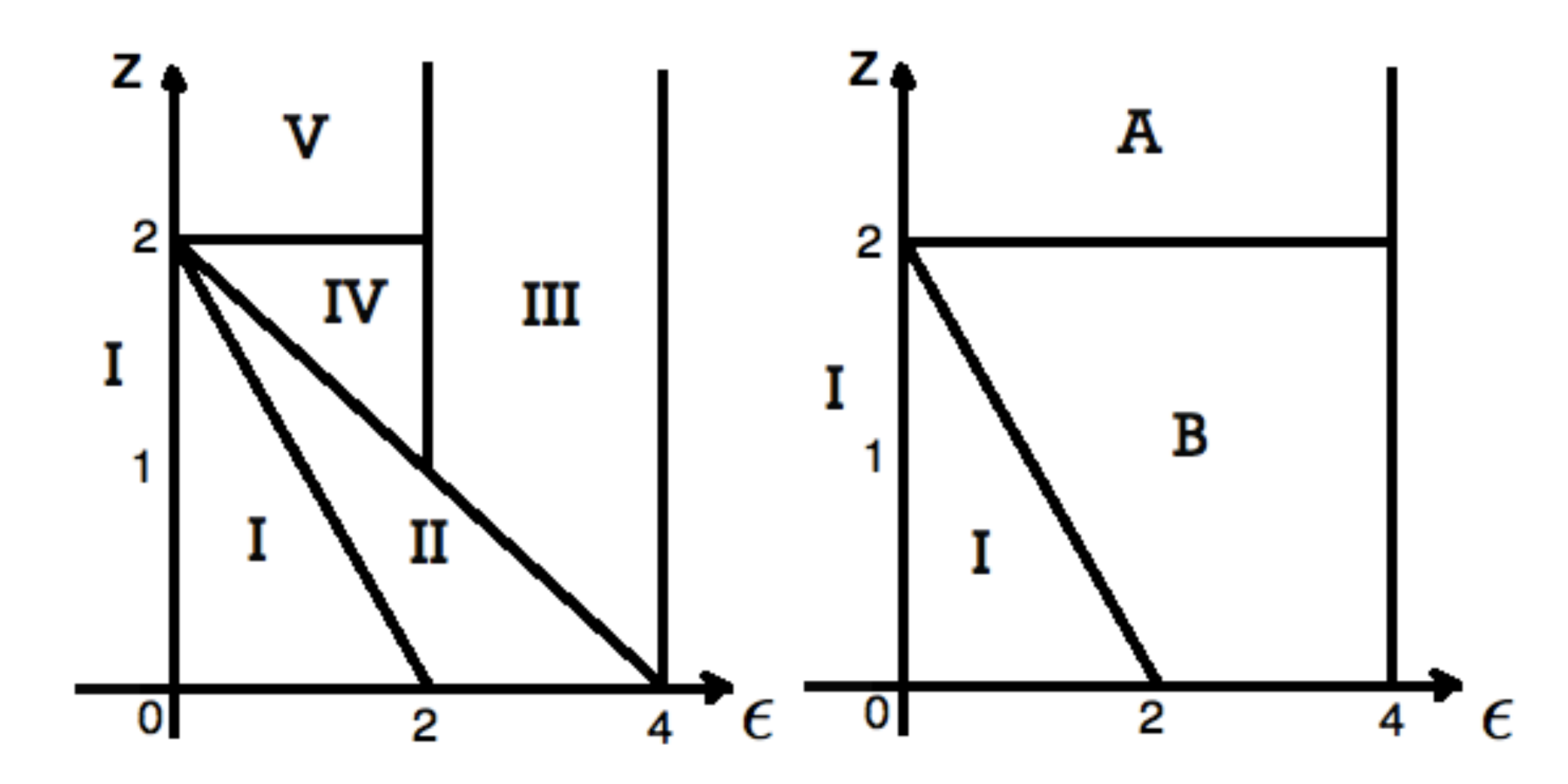}\end{center}

In this paper we show that the exact RG
method, inspired by \cite{Polchinski_Renormalization_1984,Wilson_Renormalization_1974}
is capable of recovering all the scaling regimes with correct scalings.
The method in \cite{Polchinski_Renormalization_1984,Wilson_Renormalization_1974}
was developed to study the (Euclidean) quantum field theories, which typically
consist of a Gaussian field $\phi$ on $\mathbb{R}^{d}$ with distribution
formally written as $e^{-\frac{1}{2}\int(\partial\phi(x))^{2}d^{d}x}$,
and a functional in $\phi$ that is written in a form $e^{\mathcal{V}(\phi)}$,
for instance $\mathcal{V}(\phi)=-\frac{1}{4}\int\phi(x)^{4}d^{d}x$.
An untraviolet cutoff $\Lambda_{0}$ is needed to make sense of the
functional.  For the RG scheme, one starts by  decomposing $\phi$ into two parts
corresponding to slow and fast fluctuation:
$\hat{\phi}(k)=\hat{\phi}_{<}(k)+\hat{\phi}_{>}(k)$
separated by a scale $k\sim e^{-l}\Lambda_{0}$ and average out $e^{\mathcal{V}(\phi)}$
w.r.t. $\hat{\phi}_{>}$.
The resulting density, denoted by 
$\left\langle e^{\mathcal{V}(\phi)}\right\rangle _{\phi_{>}}$,
is then  rescaled: $k\rightarrow e^{-l}k$. In the context of
quantum field theory, Polchinski \cite{Polchinski_Renormalization_1984}
was able to write down an exact evolution equation which describes
the dynamics of $\left\langle e^{\mathcal{V}(\phi)}\right\rangle _{\phi_{>}}$
as $l$ changes. 

To apply similar ideas to our model, we regard $T(x,y,t)$ as a
functional of the Gaussian field $v$ for each $x,y,t$. We then exploit
the special structure in our model to study the resulting Polchinski type of equation.

To begin with, 
we Fourier transform $T$ in $y$, to obtain a model for $\hat{T}(x,\xi,t)$,
\begin{equation}
\begin{aligned} & \frac{\partial\hat{T}}{\partial t}+iv(x,t)\xi\hat{T}=-\frac{1}{2}\nu_{0}\xi^{2}\hat{T}+\frac{1}{2}\nu_{0}\partial_{x}^{2}\hat{T}\end{aligned}
\end{equation}
This form allows us to apply the Feynman-Kac formula
to obtain
\begin{equation}
\hat{T}(x,\xi,t)=\mathbb{E}\bigg[e^{-\frac{\nu_{0}}{2}\xi^{2}t}e^{-i\xi\int_{0}^{t}v(x+\sqrt{\nu_{0}}B_{s},t-s)ds}T\big|_{t=0}(x+\sqrt{\nu_{0}}B_{t},\xi)\bigg]\label{eq:init_FK}
\end{equation}
where $B$ is a standard Brownian motion on $\mathbb{R}$ initiated at the origin
and is independent of $v$. We will always write $\mathbb{E}$
for expectation over $B$ and $\left\langle -\right\rangle $ for
expectation over $v$. 

\section{Steady case}

\subsection{The Polchinski equation}

Write
\begin{equation}
v(x)=v_{<}(x)+v_{>}(x)
\end{equation}
where 
\begin{equation}
\left\langle \left|\hat{v}_{>}(k)\right|^{2}\right\rangle =\sqrt{2\pi}1_{e^{-l}<|k|\leq1}\left|k\right|^{1-\epsilon}
\end{equation}
\begin{equation}
\left\langle \left|\hat{v}_{<}(k)\right|^{2}\right\rangle =\sqrt{2\pi}1_{\delta\leq|k|\leq e^{-l}}\left|k\right|^{1-\epsilon}
\end{equation}
Let $\bar{T}_{l}(x,\xi,t;v_{<})=\langle \hat{T}(x,\xi,t)\rangle_{v_{>}}$
be the average of $\hat{T}(x,\xi,t)$ over $v_{>}$. 

\begin{prop}
\label{prop:pol}$\bar{T}_{l}(x,\xi,t;v_{<})$ satisfies the 
equation (\cite{Polchinski_Renormalization_1984})
\begin{equation}
\frac{\partial\bar{T}_{l}(x,\xi,t;v_{<})}{\partial l}=\int\int\frac{\partial}{\partial l}C_{l}(x^{\prime}-y^{\prime})\frac{\delta^{2}\bar{T}_{l}(x,\xi,t;v_{<})}{\delta v_{<}(x^{\prime})\delta v_{<}(y^{\prime})}dx^{\prime}dy^{\prime}
\label{eq:Pol_steady}
\end{equation}
for every $x,\xi,t$, where the right hand side involves functional derivatives of $\bar{T}_{l}$
w.r.t. $v_{<}$, and 
\begin{equation}
C_{l}(z)=\int_{e^{-l}}^{1}e^{izk}\left|k\right|^{1-\epsilon}dk
\end{equation}
and the initial condition at $l=0$ is
\begin{equation}
\bar{T}_{0}(x,\xi,t;v)=\hat{T}(x,\xi,t)=\mathbb{E}\bigg[e^{-\frac{\nu_{0}}{2}\xi^{2}t}e^{-i\xi\int_{0}^{t}v(x+\sqrt{\nu_{0}}B_{s})ds}T\big|_{t=0}(x+\sqrt{\nu_{0}}B_{t},\xi)\bigg]\label{eq:init_FK-1}
\end{equation}
\end{prop}
Notice that here $\bar{T}$ is viewed as a function or functional of $l$ and $v_{<}$,
with $(x, \xi, t)$ as parameters.


\begin{proof}
Let $\mu_{l}$ be the mean zero Gaussian 
measure defined on the space of smooth functions $C^{\infty}$ with
$C_{l}$ as covariance (the mean zero condition and the covariance uniquely determined the Gaussian measure). In other words we have
\begin{equation}
\left\langle \left|\hat{v}(k)\right|^{2}\right\rangle_{\mu_l} =\sqrt{2\pi}1_{e^{-l}<|k|\leq1}\left|k\right|^{1-\epsilon}
\end{equation}
for all the random variable $C^{\infty}\ni v \mapsto \hat{v}(k)$.
Recall the elementary fact that the Gaussian density $\rho=\frac{1}{\mathcal{N}}\exp{(-\frac{1}{2}X^T c^{-1}(t) X)}$ satisfies the heat equation $\partial_t \rho=\dot{c}\Delta\rho$
where $\dot{c}=\partial_t c$ and $\mathcal{N}$ is a normalization factor so that $\int_{\{X\}}\rho=1$ (it can be easily checked by Fourier transform in $X$).
Therefore, 
\begin{equation}
\frac{\partial\mu_{l}(v)}{\partial l}=\int\int\frac{\partial}{\partial l}C_{l}(x^{\prime}-y^{\prime})\frac{\delta^{2}\mu_{l}(v)}{\delta v(x^{\prime})\delta v(y^{\prime})}dx^{\prime}dy^{\prime}
\end{equation}
Furthermore let us notice that when $l=0$, $v$ is almost surely zero w.r.t. 
$\mu_{l=0}$, i.e. $\mu_{l=0}$ is concentrated at $v=0$.
Therefore $\mu_{l}$ is the fundamental solution of the above heat equation
with $\frac{\partial}{\partial l}C_{l}(x^{\prime}-y^{\prime})$ as
coefficient of the Laplacian. 

Since  $\mu_{l}(v)=\mu_{l}(-v)$, we have
\begin{equation}
\begin{aligned}
& \left\langle \hat{T}(x,\xi,t;v)\right\rangle _{v_{>}}
=\int\hat{T}(x,\xi,t;v+u)d\mu_l(u) \\
& =\int\hat{T}(x,\xi,t;u)d\mu_l(v-u)
\end{aligned}
\end{equation}
Therefore 
$\left\langle \hat{T}(x,\xi,t;v)\right\rangle _{v_{>}}$
solves the same heat equation
\begin{equation}
\frac{\partial\bar{T}_{l}(x,\xi,t;v)}{\partial l}=\int\int\frac{\partial}{\partial l}C_{l}(x^{\prime}-y^{\prime})\frac{\delta^{2}\bar{T}_{l}(x,\xi,t;v)}{\delta v(x^{\prime})\delta v(y^{\prime})}dx^{\prime}dy^{\prime}
\end{equation}
with initial data at $l=0$ given by $\hat{T}(x,\xi,t;v)$.
Since $\bar{T}_{l}(x,\xi,t;v)$ 
actually only depends on the low modes $v_{<}$ of $v$, the Proposition is proved.
\end{proof}


Motivated by \eqref{eq:init_FK-1}, we seek solutions of \eqref{eq:Pol_steady}
in the form:
\begin{equation}
\bar{T}_{l}(x,\xi,t;v_{<})=\mathbb{E}\bigg[e^{\sum_{n\geq0}U_{n}(x,\xi,t;v_{<},B,l)}T\big|_{t=0}(x+\sqrt{\nu_{0}}B_{t},\xi)\bigg]
\end{equation}
where $U_{0}$ is independent of $v_{<}$, $U_{1}$ is linear in $v_{<}$,
etc. and $U_{n}$ is $n$-th order in $v_{<}$.
Substituting into \eqref{eq:Pol_steady} and comparing orders in $v_{<}$, 
we obtain a system of equations
\begin{equation}
\frac{\partial U_{n}}{\partial l}=\int\int\frac{\partial}{\partial l}C_{l}(x^{\prime}-y^{\prime})\left\{ \frac{\delta^{2}U_{n+2}}{\delta v_{<}(x^{\prime})\delta v_{<}(y^{\prime})}+\sum_{p+q=n+2}\frac{\delta U_{p}}{\delta v_{<}(x^{\prime})}\frac{\delta U_{q}}{\delta v_{<}(y^{\prime})}\right\} dx^{\prime}dy^{\prime}\label{eq:Un}
\end{equation}
The initial conditions at $l=0$ are:
$U_0= -\frac{\nu_{0}}{2}\xi^{2}t, U_1=-i\xi\int_{0}^{t}v(x+\sqrt{\nu_{0}}B_{s})ds,
 U_{n}=0$ for $n\geq2$.
By inspection, we have $\frac{\partial}{\partial l}U_{n}=0$ for $n\geq1$ and
$U_{n}=0$ for $n\geq2$,
\begin{equation}
U_{1}(x,\xi,t;v_{<},B,l)=-i\xi\int_{0}^{t}v_{<}(x+\sqrt{\nu_{0}}B_{s})ds\label{eq:U1}
\end{equation}
for all $l\geq0$. 
$U_0$ satisfies
\begin{equation}
\frac{\partial U_{0}}{\partial l}=-\xi^{2}\int\int\frac{\partial}{\partial l}C_{l}(x^{\prime}-y^{\prime})\frac{\delta\int_{0}^{t}v_{<}(x+\sqrt{\nu_{0}}B_{s})ds}{\delta v_{<}(x^{\prime})}\frac{\delta\int_{0}^{t}v_{<}(x+\sqrt{\nu_{0}}B_{s})ds}{\delta v_{<}(y^{\prime})}dx^{\prime}dy^{\prime}\label{eq:dUdllong}
\end{equation}
i.e.
\begin{equation}
\frac{\partial U_{0}}{\partial l}=-\xi^{2}\int_{0}^{t}\int_{0}^{t}\frac{\partial}{\partial l}C_{l}(\sqrt{\nu_{0}}(B_{s}-B_{s^{\prime}}))dsds^{\prime}\label{eq:dUdlshort}
\end{equation}

Next, we come to the rescaling step.
Formally, we can view the distribution of $\bar{T}_{l}(x,\xi,t)$
as superpositions (over Brownian paths) of quantum field theoretic measures:
\begin{equation}
\exp\left\{ -\frac{1}{2}\int\left(v_{<}(x^{\prime})(-i\partial)^{\epsilon-1}v_{<}(x^{\prime})\right)dx^{\prime}+U_{1}(l)+U_{0}(l)\right\} 
\end{equation}
with
\begin{equation}
\frac{1}{2}\int\left(v(x^{\prime})(-i\partial)^{\epsilon-1}v(x^{\prime})\right)dx^{\prime}=\int_{\delta}^{e^{-l}}|k|^{\epsilon-1}|\hat{v}_{<}(k)|^{2}dk\label{eq:quadratic}
\end{equation}
With the rescaling $x\rightarrow e^{l}x$, $t\rightarrow e^{\alpha l}t$,
the scaling exponent for $v$ can be found by requiring that the quadratic term be preserved.
This gives 
\begin{equation}
v_{<}(e^{l}x)\rightarrow e^{(\epsilon/2-1)l}v(x)
\end{equation}

Let
\begin{equation}
V_{1}(x,\xi,t;v,B,l)=U_{1}(e^{l}x,e^{-l}\xi,e^{\alpha l}t;e^{(\epsilon/2-1)l}v_{<}(e^{-l}\cdot),B,l)\label{eq:defV1}
\end{equation}
\begin{equation}
V_{0}(x,\xi,t;B,l)=U_{0}(e^{l}x,e^{-l}\xi,e^{\alpha l}t;B,l)+\frac{\nu_{0}}{2}e^{(\alpha-2)l}\xi^{2}t\label{eq:defV0}
\end{equation}
Note that we have separated out the initial condition $-\frac{\nu_{0}}{2}\xi^{2}t$
from $U_{0}$, as a matter of convenience. From (\ref{eq:U1}) (\ref{eq:dUdlshort}),
we obtain  the RG flow equation after rescaling is incorporated:
\begin{equation}
\begin{cases}
\frac{\partial V_{1}}{\partial l}=(\alpha+\frac{\epsilon}{2}-2)V_{1}+(\frac{\alpha}{2}-1)\int\frac{\delta V_{1}}{\delta B_{s}}B_{s}ds\\
\frac{\partial V_{0}}{\partial l}
=(2\alpha-2)V_{0}
-e^{(2\alpha-2)l}\xi^{2}\int_{0}^{t}\int_{0}^{t}
\frac{\partial}{\partial l}\bigg[C_{l}\left(\sqrt{\nu_0}(B_{s}-B_{s^{\prime}})e^{\frac{\alpha}{2}l}\right)\bigg]dsds^{\prime}
\end{cases}\label{eq:U0U1}
\end{equation}

\subsection{Fixed points of the RG flow}

Before discussing fixed points for the RG flow, we make a
general remark concerning the infrared cut-off. 

The infrared cut-off $\delta$ is a technical device needed in order to
guarantee that the stochastic PDE for the passive scalar is well defined.
However, it prevents us from being able to take the long time
limit for the RG flow. In fact, with the infrared cut-off,
the RG flow stops at the scale $l=-\log\delta$. 

Nevertheless, we can still study the fixed points in the following sense. We started from the dynamical system (\ref{eq:Pol_steady}), which stops at the scale $l=-\log\delta$ due to the infrared cut-off as discussed above, and then we have reduced (\ref{eq:Pol_steady}) to a dynamical system of the form
(\ref{eq:U0U1}).
One of the advantages among others of this reduction is that the dynamical system 
(\ref{eq:U0U1}) for $V_1,V_0$ doesn't really need an infrared cut-off in order to be well-defined.
In other words,
(\ref{eq:U0U1}) itself actually exists as $l\rightarrow\infty$. Therefore,
suppose that $(V_{1}^{\star},V_{0}^{\star})$ is
a fixed point and $(V_{1}(l),V_{0}(l))\rightarrow(V_{1}^{\star},V_{0}^{\star})$
as $l\rightarrow\infty$,
we will have (under certain suitable norm) $\left\Vert (V_{1}(l),V_{0}(l))-(V_{1}^{\star},V_{0}^{\star})\right\Vert <a(l)$
where $a(l)\rightarrow0$ as $l\rightarrow\infty$. Consequently we will obviously have
\begin{equation}
\left\Vert (V_{1}(-\log\delta),V_{0}(-\log\delta))-(V_{1}^{\star},V_{0}^{\star})\right\Vert <a(-\log\delta)
\end{equation}
Once we have this bound, we will then go back to (\ref{eq:Pol_steady}) and have estimate on
$\bar{T}_l$ at $l=-\log\delta$.
We will then let $\delta\rightarrow0$ (see
(\ref{eq:main_question})).


\subsubsection{The mean field regime $\epsilon<0$}

With diffusive scaling, i.e. $\alpha=2$, we see that 
\begin{equation}
\frac{\partial V_{1}}{\partial l}=\frac{\epsilon}{2}V_{1}
\end{equation}
and this together with $\epsilon<0$ implies $V_{1}\rightarrow0$.
This implies that in the mean field regime, the fixed points are
\begin{equation}
e^{-\frac{1}{2}\int\left(v(x^{\prime})(-i\partial)^{\epsilon-1}
v(x^{\prime})\right)dx^{\prime}+V_{0}}
\end{equation}
for all dimensionless function $V_{0}$, for instance $V_{0}=D\xi^{2}t$,
$D\in\mathbb{R}$. 

Next we ask: Which specific fixed point (i.e. $V_{0}=?$) the RG flow
converges to starting from $U_{1},U_{0}$? 
Solving equation (\ref{eq:dUdlshort})
with $U_{0}(l=0)=-\frac{\nu_{0}}{2}\xi^{2}t$,
we obtain
\begin{equation}
U_{0}(l)=-\frac{\nu_{0}}{2}\xi^{2}t-\xi^{2}\int_{0}^{t}\int_{0}^{t}C_{l}(\sqrt{\nu_{0}}(B_{s}-B_{s^{\prime}}))dsds^{\prime}\label{eq:ergodU}
\end{equation}
Therefore
\begin{equation}
V_{0}(l)=-e^{-2l}\xi^{2}\int_{0}^{e^{2l}t}\int_{0}^{e^{2l}t}C_{l}(\sqrt{\nu_{0}}(B_{s}-B_{s^{\prime}}))dsds^{\prime}\label{eq:ergod}
\end{equation}
Using ergodicity
arguments, the last term in the above equation converges to 
\begin{equation}
\frac{2}{\pi\nu_{0}}t\xi^{2}\int_{0}^{1}|k|^{-1-\epsilon}dk
=-\frac{2t\xi^{2}}{\pi\nu_{0}\epsilon}
\end{equation}
as $l\rightarrow\infty$. In fact, following \cite{Avellaneda_Mathematical_1990},
one can construct a process $X$ such that $X_{>}^{\prime\prime}=-\frac{2}{\nu_{0}}v_{>}$,
and by Ito's formula, $\xi\int_{0}^{t}v_{>}(\sqrt{\nu_{0}}B_{s})ds=\sqrt{\nu_{0}}\xi\int_{0}^{t}X_{>}^{\prime}(\sqrt{\nu_{0}}B_{s})dB_{s}+R_{>}(t)$
where it can be shown that $R_{>}(t)$ upon rescaling goes to $0$
as $l\rightarrow\infty$. The Ito integral is a martingale with quadratic
variation $\nu_{0}\xi^{2}\int_{0}^{t}X_{>}^{\prime}(\sqrt{\nu_{0}}B_{s})^{2}ds$,
which upon rescaling and by ergodicity theorem goes to $\nu_{0}\xi^{2}t\left\langle X^{\prime}(0)^{2}\right\rangle =\frac{2}{\pi\nu_{0}}t\xi^{2}\int_{0}^{1}|k|^{-1-\epsilon}dk$.
The claim above follows by observing that the correlation
of $\xi\int_{0}^{t}v_{>}(\sqrt{\nu_{0}}B_{s})ds$ is given by the
last term of (\ref{eq:ergodU}).

This argument also gives an effective equation  for $\bar{T}$:
\begin{equation}
\frac{\partial\bar{T}}{\partial t}
=\frac{1}{2}\nu_{0}\Delta\bar{T}-\frac{2}{\pi\nu_{0}\epsilon}\frac{\partial^2\bar{T}}{\partial y^2}
\end{equation}
Since $\epsilon<0$ the coefficient of $\partial^2/\partial y^2$ is enhanced.

\subsubsection{Fixed points with non-vanishing limiting $V_{1}$: 
the regime $2<\epsilon<4$}

We have seen that in the mean field regime, $V_{1}^{\star}=0$ at the fixed point.
We now look for a scaling such that the fixed points have linear terms in
$v$ (i.e. $V_{1}^{\star}\neq0$), for $\epsilon>0$. From (\ref{eq:U0U1}),
$\partial_l V_{1}^{\star}=0$ can be guaranteed by
\begin{equation}
\alpha=2-\frac{\epsilon}{2}
\end{equation}
and 
\begin{equation}
\int\frac{\delta V_{1}^{\star}}{\delta B_{s}}B_{s}ds=0
\end{equation}
The second condition implies that $V_{1}^{\star}$ does not depend on $B$. In fact
\begin{equation}
V_{1}^{\star}=-i\xi\int_{0}^{t}v(x)ds=-i\xi tv(x)
\end{equation}

Now with $V_{1}^{\star}$ already found, we identify the constant
term in the fixed point, namely $V_{0}^{\star}$. By (\ref{eq:U0U1})
with $B=0$
\begin{equation}
\frac{\partial V_{0}}{\partial l}=(2\alpha-2)V_{0}-\frac{1}{\pi}e^{(2\alpha-2)l}\xi^{2}\int_{0}^{t}\int_{0}^{t}\frac{\partial}{\partial l}\int_{e^{-l}}^{1}\left|k\right|^{1-\epsilon}dkdsds^{\prime}
\end{equation}
The right hand side being equal to $0$ implies that
\begin{equation}
(2\alpha-2)V_{0}^{\star}-\frac{1}{\pi}e^{(2\alpha-2)l}\xi^{2}\int_{0}^{t}\int_{0}^{t}e^{-l(2-\epsilon)}dsds^{\prime}=0
\end{equation}
Using the scaling we found above $2\alpha-2=2-\epsilon$, we get
\begin{equation}
(2\alpha-2)V_{0}^{\star}-\frac{1}{\pi}\xi^{2}t^{2}=0
\end{equation}
If $\epsilon > 2$, then we have
\begin{equation}
V_{0}^{\star}=\frac{1}{\pi}\frac{\xi^{2}t^{2}}{2-\epsilon}=-\frac{1}{\pi}\xi^{2}t^{2}\int_{1}^{\infty}\left|k\right|^{1-\epsilon}dk
\end{equation}
This is the {\it hyperscaling regime} identified in \cite{Avellaneda_Mathematical_1990}.
Note that the term $-\frac{\nu_{0}}{2}\xi^{2}t$
in $U_{0}$ goes to zero under the scaling $\alpha=2-\frac{\epsilon}{2}$.

For $0<\epsilon<2$, $\frac{1}{\pi}\frac{\xi^{2}t^{2}}{2-\epsilon}$
is not the fixed point that the RG flow converges to, starting
from our initial condition $V_{1}(0),V_{0}(0)$.
To see this, we examine the effects of integration and of rescaling separately.
Integrating out $e^{-l}<|k|<1$ gives an increment
\begin{equation}
\begin{aligned}
-\frac{1}{\pi}\xi^{2}t^{2}\int_{1}^{\infty}\left|k\right|^{1-\epsilon}dk
&\rightarrow-\frac{1}{\pi}\xi^{2}t^{2}\int_{1}^{\infty}\left|k\right|^{1-\epsilon}dk-\frac{1}{\pi}\xi^{2}t^{2}\int_{e^{-l}}^{1}\left|k\right|^{1-\epsilon}dk\\
&=-\frac{1}{\pi}\xi^{2}t^{2}\int_{e^{-l}}^{\infty}\left|k\right|^{1-\epsilon}dk
\end{aligned}
\end{equation}
Rescaling under the change of variable $k\rightarrow e^{-l}k$ gives 
\begin{equation}
-\frac{1}{\pi}e^{(2\alpha-2)l}\xi^{2}t^{2}\int_{1}^{\infty}\left|k\right|^{1-\epsilon}e^{(\epsilon-2)l}dk=-\frac{1}{\pi}\xi^{2}t^{2}\int_{1}^{\infty}\left|k\right|^{1-\epsilon}dk
\end{equation}
If $0<\epsilon<2$, solving (\ref{eq:dUdllong})(\ref{eq:dUdlshort})
results in
\begin{equation}
U_{0}(l)=U_{0}(0)-\frac{1}{\pi}\xi^{2}\int_{0}^{t}\int_{0}^{t}\int_{e^{-l}}^{1}\left|k\right|^{1-\epsilon}e^{ik(B_{s}-B_{s^{\prime}})}dkdsds^{\prime}\label{eq:hyperU0}
\end{equation}
with $U_{0}(0)=-\frac{\nu_{0}}{2}\xi^{2}t$. Upon rescaling,
the effect of the Brownian motion disappears and we obtain
$V_{0}(l)\rightarrow-\frac{1}{\pi}\xi^{2}t^{2}\int_{1}^{\infty}\left|k\right|^{1-\epsilon}dk=-\infty$
for $0<\epsilon<2$.

Unlike the mean field regime where the constant term in $v$ 
can be an arbitrary dimensionless function at the fixed point, here the linear term determines
uniquely the constant term. Indeed in the
mean field regime, which fixed point the RG flow converges to depends on the initial
condition ($\nu_{0}$), but in the hyperscaling regime $2<\epsilon<4$,
the initial condition ($\nu_{0}$) does not affect the infrared behavior.

To find the effective equation for $\bar{T}$ (see (\ref{eq:main_question})), we
first take an infrared cutoff $\delta$ for $v$, run the RG flow until
$e^{-l}=\delta$ to obtain
\begin{equation}
V_{0}(e^{-l}=\delta)
=-\frac{1}{\pi}\xi^{2}t^{2}\int_{1}^{e^{l}}\left|k\right|^{1-\epsilon}dk +o(1)
\end{equation}
which is very close to $V_{0}^{\star}$.
Here the $o(1)$ term is a correction term due to the effect of the Brownian motion.
Then we take $\delta\rightarrow0$.
In this way, we recover the effective model derived in \cite{Avellaneda_Mathematical_1990}:
\begin{equation}
\frac{\partial \bar{T}}{\partial t} =   \frac{t}{\pi} \int_{1}^{\infty}\left|k\right|^{1-\epsilon}dk \frac{\partial^2 \bar{T}}{\partial y^2} 
\end{equation}


\subsubsection{The nonlocal regime }

Now we consider the case when  $0<\epsilon<2$.  We have
\begin{equation}
U_{0}(l)=U_{0}(0)
-\frac{1}{\pi}\xi^{2}t^{2}
\int_{0}^{1}\int_{0}^{1}\int_{e^{-l}}^{1}
e^{i(B_{s}-B_{s^{\prime}})\sqrt{\nu_0 t}k}\left|k\right|^{1-\epsilon}dkdsds^{\prime}
\end{equation}
with $U_{0}(0)=-\frac{\nu_{0}}{2}\xi^{2}t$. The last term under rescaling
$x\rightarrow e^{l}x$, $t\rightarrow e^{\alpha l}t$ becomes 
\begin{equation}
-\frac{1}{\pi}\xi^{2}t^{2}e^{2(\alpha-1)l}
\int_{0}^{1}\int_{0}^{1}\int_{e^{-l}}^{1}e^{i(B_{s}-B_{s^{\prime}})e^{\alpha l/2}\sqrt{\nu_0 t}k}\left|k\right|^{1-\epsilon}dkdsds^{\prime}
\end{equation}
With a change of variable for $k$, $k\rightarrow e^{-\alpha l/2}(\nu_0 t)^{-1/2}k$,  we have
\begin{equation}
\begin{aligned} 
& -\frac{1}{\pi}e^{(\alpha+\alpha\epsilon/2-2)l}\xi^{2}\nu_0^{\epsilon_0/2-1}t^{1+\epsilon/2}\int_{0}^{1}\int_{0}^{1}\int_{e^{(\alpha/2-1)l}\sqrt{t}}^{e^{\alpha l/2}\sqrt{t}}e^{i(B_{s}-B_{s^{\prime}})k}\left|k\right|^{1-\epsilon}dkdsds^{\prime}
\end{aligned}
\end{equation}
Let $\alpha=\frac{2}{1+\epsilon/2}$, we obtain a nontrivial
fixed point as $l\rightarrow\infty$
\begin{equation}\label{steady_V0ast}
V_{0}^{\star}=-\frac{1}{\pi}\xi^{2}\nu_0^{\epsilon_0/2-1}t^{1+\epsilon/2}\int_{0}^{1}\int_{0}^{1}\int e^{i(B_{s}-B_{s^{\prime}})k}\left|k\right|^{1-\epsilon}dkdsds^{\prime}
\end{equation}
Notice that $U_{0}(0)=-\frac{\nu_{0}}{2}\xi^{2}t$ vanishes in the limit. From \cite{Avellaneda_Mathematical_1990}, Proposition 4.1, we know that the right hand
side of the above equation is well-defined a.s. with respect to Wiener measure of $B$.

In this regime, unlike the mean field or the hyperscaling regime discussed above, now the fixed point $V_0^{\ast}$ depends on the Brownion motion $B$. The averaged quantity 
\begin{equation}\label{steady_nonlocalTbar}
\bar{T}(x,\xi,t)=\mathbb{E}\bigg[
e^{V_0^{\ast}(x,\xi,t;B)}
T\big|_{t=0}(x+\sqrt{\nu_{0}}B_{t},\xi)\bigg]
\end{equation}
will satisfy a non-local effect PDE. For a more general discussion on this issue, see Prop \ref{prop:nonlocal_effPDE} in Section \ref{sec:Effective-SPDE}. 
We  observe (\ref{steady_V0ast}) (\ref{steady_nonlocalTbar}): if 
\begin{equation}\label{eq:steady_alpha}
-\frac{1}{\pi}\int_{0}^{1}\int_{0}^{1}\int e^{i(B_{s}-B_{s^{\prime}})k}\left|k\right|^{1-\epsilon}dkdsds^{\prime}=\alpha
\end{equation}
we would have an effective equation
\begin{equation}
\frac{\partial \bar{T}_{\alpha}}{\partial t} = D(\nu_0) \alpha t^{\epsilon/2}  \frac{\partial^2 \bar{T}_{\alpha}}{\partial y^2} 
\end{equation}
for some constant $D(\nu_0)$. Here $\alpha$ depends on $B$, so one has to define a measure $\nu(\alpha)$ on $\mathbb{R}$ as the push-forward measure of the Wiener measure via (\ref{eq:steady_alpha}), and then it holds that
\begin{equation}
\bar{T}(x,y,t)=\int\bar{T}_{\alpha}(x,y,t)d\nu(\alpha)
\end{equation}
In other words 
\begin{equation}
\bar{T}(x,\xi,t)=K(\xi,t)\hat{T}\big|_{t=0}(x,\xi)
\end{equation}
where $K=\int K_{\alpha} d\nu(\alpha)$ and $K_{\alpha}$ is the heat kernel for a heat equation with coeffecient $D\alpha t^{\epsilon/2}$.
It is easy to see that the nonlocal equation for $\bar{T}$ is a special case of 
the nonlocal effective equation in Prop \ref{prop:nonlocal_effPDE} in 
Section \ref{sec:Effective-SPDE}. 


\begin{rem}\label{steady_sum}
The full pictures for the RG flow in the steady case can be summarized as follows. 
Let $\mathcal{E}$ be the space of all stochastic equations of the form
\begin{equation}
\frac{\partial T}{\partial t} = F(T,v)
\end{equation}
where $v$ is a Gaussian noise. The behavior of RG flow depends on the spectrum of $v$. The {\it local} equations form a subspace $\mathcal{L}\subset\mathcal{E}$. In $\mathcal{L}$ there is a one-dimensional subspace $\mathcal{A}\subset\mathcal{L}$ parametrized by $\nu_0$ whose elements are SPDE's of the form (\ref{eq:AMmodel}). Our initial data of the RG flow is always taken from $\mathcal{A}$. As will be seen from Prop \ref{prop:nonlocal_effPDE} in Section \ref{sec:Effective-SPDE}, the RG flow starting from $\mathcal{A}\subset\mathcal{L}$ will immediately exit $\mathcal{L}$. In the first two regimes discussed above, the RG flow eventually converges to a fixed point in $\mathcal{L}$, but in the nonlocal regime, the RG flow converges to a fixed point outside $\mathcal{L}$. Furthermore, the three regimes have another interesting difference on the behavior of the RG dynamics: 
for the mean field and nonlocal regimes,
different points of $\mathcal{A}$ belong to the basins of attraction of different 
fixed points regimes since the fixed point that the RG flow converges to 
depends on $\nu_0$, while for the hyperscaling regime all points of $\mathcal{A}$ 
are in the same basin of attraction (i.e. the same "universality class").
\end{rem}

\section{The unsteady case}

\subsection{The Polchinski equation}

As in the steady case, we decompose the Gaussian field 
\begin{equation}
v(x,t)=v_{<}(x,t)+v_{>}(x,t)
\end{equation}
where 
\begin{equation}
\left\langle \left|\hat{v}_{>}(k,\omega)\right|^{2}\right\rangle =\sqrt{2\pi}1_{e^{-l}\leq|k|\leq1}|k|^{1-\epsilon}\frac{|k|^{z}}{|k|^{2z}+\omega^{2}}
\end{equation}
\begin{equation}
\left\langle \left|\hat{v}_{<}(k,\omega)\right|^{2}\right\rangle =\sqrt{2\pi}1_{\delta\leq|k|\leq e^{-l}}|k|^{1-\epsilon}\frac{|k|^{z}}{|k|^{2z}+\omega^{2}}
\end{equation}
Repeating the proof of Prop \ref{prop:pol}, we obtain the analogous Polchinski type of
equation for $\bar{T}_{l}(x,\xi,t,v_{<})$ 
\begin{equation}
\frac{\partial\bar{T}_{l}(x,\xi,t;v_{<})}{\partial l}=\int\int\frac{\partial}{\partial l}C_{l}(x^{\prime}-y^{\prime},t^{\prime}-r^{\prime})\frac{\delta^{2}\bar{T}_{l}(x,\xi,t;v_{<})}{\delta v_{<}(x^{\prime},t^{\prime})\delta v_{<}(y^{\prime},r^{\prime})}dx^{\prime}dy^{\prime}dt^{\prime}dr^{\prime}\label{eq:Pol_unsteady}
\end{equation}
where
\begin{equation}
C_{l}(z, \tau)=\int_{-\infty}^{\infty}\int_{e^{-l}}^{1}e^{izk+i\tau\omega}
\left|k\right|^{1-\epsilon}\frac{|k|^{z}}{|k|^{2z}+\omega^{2}}dkd\omega
\end{equation}
The initial condition at $l=0$ is
\begin{equation}
\bar{T}_{0}(x,\xi,t;v)=\hat{T}(x,\xi,t)=\mathbb{E}\bigg[e^{-\frac{\nu_{0}}{2}\xi^{2}t}e^{-i\xi\int_{0}^{t}v(x+\sqrt{\nu_{0}}B_{s},t-s)ds}T\big|_{t=0}(x+\sqrt{\nu_{0}}B_{t},\xi)\bigg]\label{eq:init_FK-1-1}
\end{equation}
Following the procedure in the steady case, we write the system of
equations for $U_{n}, n \ge 0$.  Again we find that $U_{n}=0$ for $n\geq2$ and
\begin{equation}
U_{1}=-i\xi\int_{0}^{t}v(x+\sqrt{\nu_{0}}B_{s},t-s)ds
\end{equation}
for all $l\geq0$. The RG flow for $U_{0}$ is
\begin{equation}
\begin{aligned} & \frac{\partial U_{0}}{\partial l}=-\xi^{2}\int\int\int\int\frac{\partial}{\partial l}C_{l}(x^{\prime}-y^{\prime},t^{\prime}-r^{\prime})\\
 & \qquad\frac{\delta\int_{0}^{t}v_{<}(x+\sqrt{\nu_{0}}B_{s},t-s)ds}{\delta v_{<}(x^{\prime},t^{\prime})}\frac{\delta\int_{0}^{t}v_{<}(x+\sqrt{\nu_{0}}B_{s},t-s)ds}{\delta v_{<}(y^{\prime},r^{\prime})}dx^{\prime}dy^{\prime}dt^{\prime}dr^{\prime}
\end{aligned}
\label{eq:dUdllong-1}
\end{equation}
i.e.
\begin{equation}
\frac{\partial U_{0}}{\partial l}=-\xi^{2}\int_{0}^{t}\int_{0}^{t}\frac{\partial}{\partial l}C_{l}(B_{s}-B_{s^{\prime}},s-s^{\prime})dsds^{\prime}\label{eq:dUdlshort-1}
\end{equation}

In order to find the rescaling exponent for $v$ we identify the free
propagator in the unsteady case:
\begin{equation}
\int\int v(-i\partial_{x})^{\epsilon-1-z}((-i\partial_{x})^{2z}+(-i\partial_{t})^{2})vdxdt=\int\int_{\delta\leq|k|\leq1}|k|^{\epsilon-1-z}(|k|^{2z}+\omega^{2})|\hat{v}|^{2}dkd\omega
\end{equation}
There are two terms. We will discuss different cases in which different
term dominates. In the integration step, the two terms together plays
the role of the free propagator. In the rescaling step, since in
general the two terms have different dimensions, we have to rescale
in such a way that only one of the two terms is invariant and the other term damps out.
For this reason we define the rescaled velocity with scaling exponent
of $v$ yet to be found: 
\begin{equation}
V_{1}(x,\xi,t;v,B)=U_{1}(e^{l}x,e^{-l}\xi,e^{\alpha l}t;e^{[v]l}v_{<}(e^{-l}\cdot,e^{-\alpha l}t),B)
\end{equation}
\begin{equation}
V_{0}(x,\xi,t;B)=U_{0}(e^{l}x,e^{-l}\xi,e^{\alpha l}t;B)+\frac{\nu_{0}}{2}e^{(\alpha-2)l}\xi^{2}t\label{eq:defV0-1}
\end{equation}
The exact RG flow with rescaling incorporated is then defined by:
\begin{equation}
\begin{cases}
\frac{\partial V_{1}}{\partial l}=(\alpha-1+[v])V_{1}+(\frac{\alpha}{2}-1)\int\frac{\partial V_{1}}{\partial B_{s}}B_{s}ds\\
\frac{\partial V_{0}}{\partial l}=(2\alpha-2)V_{0}-e^{(2\alpha-2)l}\xi^{2}\int_{0}^{t}\int_{0}^{t}\frac{\partial}{\partial l}\bigg[C_{l}\left((B_{s}-B_{s^{\prime}})e^{\frac{\alpha}{2}l},s-s^{\prime}\right)\bigg]dsds^{\prime}
\end{cases}\label{eq:U0U1-1}
\end{equation}
with initial condition at $l=0$ givin by 
\begin{equation}
V_{1}(l=0)=-i\xi\int_{0}^{t}v(x+\sqrt{\nu_{0}}B_{s},t-s)ds\qquad V_{0}(l=0)=0
\end{equation}

\subsection{Fixed points for the unsteady case}

\subsubsection{Regime I (mean field) $\{\epsilon<0,z\geq2\}\cup\{\epsilon<2-z,0<z<2\}$}

With diffusive scaling, i.e. $\alpha=2$, we get
\begin{equation}
\frac{\partial V_{1}}{\partial l}=(1+[v])V_{1}\label{eq:RegIV1}
\end{equation}
There are two separate cases to be handled. In the case $z\geq2$, $\omega^{2}$
dominates $k^{2z}+\omega^{2}$, therefore $[v]$ is determined by requiring that 
\begin{equation}
\int\int v(-i\partial_{x})^{\epsilon-1-z}(-i\partial_{t})^{2}vdxdt
\end{equation}
be invariant, which implies
\begin{equation}
[v]=\frac{\epsilon-z}{2}
\end{equation}
It is easy to see that $1+[v]<0$ if $\epsilon<0$, hence from (\ref{eq:RegIV1}),
$V_{1}\rightarrow0$. 

In the case $0<z<2$, $k^{2z}$ dominates $k^{2z}+\omega^{2}$.  Therefore
$[v]$ is determined by requiring that 
\begin{equation}
\int\int v(-i\partial_{x})^{\epsilon-1-z}(-i\partial_{x})^{2z}vdxdt
\end{equation}
be invariant, which implies
\begin{equation}
[v]=\frac{\epsilon+z}{2}-2
\end{equation}
We still have $1+[v]<0$ if $\epsilon<2-z$. So by (\ref{eq:RegIV1}),
$V_{1}\rightarrow0$. 

As in the steady case, all dimensionless $V_{0}$ are fixed points.
To find out which specific fixed point the RG flow converges to
starting from $V_{1}(l=0),V_{0}(l=0)$, we solve  (\ref{eq:dUdlshort-1})
with $U_{0}(l=0)=-\frac{\nu_{0}}{2}\xi^{2}t$ to get,
\begin{equation}
U_{0}(l)=-\frac{\nu_{0}}{2}\xi^{2}t-\xi^{2}\int_{0}^{t}\int_{0}^{t}C_{l}(\sqrt{\nu_{0}}(B_{s}-B_{s^{\prime}}),s-s^{\prime})dsds^{\prime}\label{eq:ergodU-1}
\end{equation}
Therefore
\begin{equation}
V_{0}(l)=-e^{-2l}\xi^{2}\int_{0}^{e^{2l}t}\int_{0}^{e^{2l}t}C_{l}(\sqrt{\nu_{0}}(B_{s}-B_{s^{\prime}}),s-s^{\prime})dsds^{\prime}\label{eq:ergod-1}
\end{equation}
As in \cite{Avellaneda_Mathematical_1990}, using ergodicity
arguments,  we see that the right hand side of the above equation goes to 
\begin{equation}
-D(\epsilon,z)=-\frac{2}{\pi}t\xi^{2}\int_{0}^{1}(\frac{\nu_{0}}{2}|k|^{2}+|k|^{z})^{-1}|k|^{1-\epsilon}dk
\end{equation}
as $l\rightarrow\infty$. 

Therefore, the effective equation is
\begin{equation}
\frac{\partial\bar{T}}{\partial t}=\frac{1}{2}\nu_{0}\Delta\bar{T}+D(\epsilon,z)\bar{T}_{yy}
\end{equation}

\subsubsection{Regime II: $2-z<\epsilon<4-2z$}

We look for fixed points with nonzero linear terms.  Assume that $\partial_{x}^{2z}$
is the dominant term in the expression $\partial_{x}^{2z}+\partial_{t}^{2}$. 
With the rescaling $x\rightarrow e^{l}x$,
$t\rightarrow e^{\alpha}t$, we rescale $v$ in such a way that the quadratic term
\begin{equation}
\int\int\left(v_{<}(x^{\prime})\partial_{x}^{\epsilon-1-z}\partial_{x}^{2z}v_{<}(x^{\prime})\right)dx^{\prime}dt
\end{equation}
is preserved, i.e. $[v]=(\epsilon+z-\alpha-2)l/2$:
\begin{equation}
v_{<}(e^{l}x,e^{\alpha l}t)\rightarrow e^{(\epsilon+z-\alpha-2)l/2}v(x,t)
\end{equation}
From (\ref{eq:U0U1-1}), $\partial_l V_{1}^{\star}=0$
can be guaranteed by $\alpha=[v]+1$ i.e. 
\begin{equation}
\alpha=4-\epsilon-z
\end{equation}
and 
\begin{equation}
\int\frac{\partial V_{1}^{\star}}{\partial B_{s}}B_{s}ds=0
\end{equation}
which implies that $V_{1}^{\star}$ does not depend on $B$. In fact
\begin{equation}
V_{1}^{\star}=-i\xi\int_{0}^{t}v(x,t-s)ds
\end{equation}

Now with $V_{1}^{\star}$ already found, we identify the constant
term in the fixed point, namely $V_{0}^{\star}$. From (\ref{eq:U0U1-1})
with $B=0$, we obtain
\begin{equation}
\begin{aligned}
\frac{\partial V_{0}}{\partial l}=(2\alpha-2)V_{0} & -  \frac{1}{\pi}e^{(2\alpha-2)l}\xi^{2}\cdot\\
& \cdot \int_{0}^{t}\int_{0}^{t}\frac{\partial}{\partial l}\int_{-\infty}^{\infty}\int_{e^{-l}}^{1}\left|k\right|^{1-\epsilon}e^{i\omega(s-s^{\prime})e^{\alpha l}}\frac{|k|^{z}}{\omega^{2}+|k|^{2z}}dkd\omega dsds^{\prime}
\end{aligned}
\end{equation}
Straightforward calculations give,
\begin{equation}
\begin{aligned}G(k,t;l) & :=\int_{0}^{t}\int_{0}^{t}\int_{-\infty}^{\infty}e^{i\omega(s-s^{\prime})e^{\alpha l}}\frac{|k|^{z}}{\omega^{2}+|k|^{2z}}d\omega dsds^{\prime}\\
 & =t^{2}\left[\frac{1}{|k|^{z}te^{\alpha l}}-\frac{1}{(|k|^{z}te^{\alpha l})^{2}}(1-e^{-|k|^{z}te^{\alpha l}})\right]
\end{aligned}
\label{eq:G}
\end{equation}
We have $|k|^{z}t^{-1}e^{\alpha l}G(k,t;l)\rightarrow1$ as $l\rightarrow\infty$.
Replacing $G(k,t;l)$ by $\frac{t}{|k|^{z}e^{\alpha l}}$, we get
\begin{equation}
\frac{\partial V_{0}}{\partial l}=(2\alpha-2)V_{0}-\frac{1}{\pi}e^{(2\alpha-2)l}\xi^{2}t\frac{\partial}{\partial l}\int_{e^{-l}}^{1}\left|k\right|^{1-\epsilon-z}e^{-\alpha l}dk\label{eq:replacing_G}
\end{equation}
The right hand side being equal to $0$ implies that
\begin{equation}
(2\alpha-2)V_{0}^{\star}-\frac{1}{\pi}e^{(2\alpha-2)l}\xi^{2}te^{-l(2-\epsilon-z)}e^{-\alpha l}-\frac{1}{\pi}e^{(2\alpha-2)l}\xi^{2}t\int_{e^{-l}}^{1}\left|k\right|^{1-\epsilon-z}(-\alpha)e^{-\alpha l}dk=0
\end{equation}
Observe that the third term divided by $(-\alpha)$ solves (\ref{eq:replacing_G}),
therefore we obtain the fixed point equation 
\begin{equation}
(2\alpha-2)V_{0}^{\star}-\frac{1}{\pi}e^{(2\alpha-2)l}\xi^{2}te^{-l(2-\epsilon-z)}e^{-\alpha l}-\alpha V_{0}^{\star}=0
\end{equation}
Using the scaling we found above $\alpha=4-\epsilon-z$, we obtain
\begin{equation}
(\alpha-2)V_{0}^{\star}-\frac{1}{\pi}\xi^{2}t=0
\end{equation}
namely, 
\begin{equation}
V_{0}^{\star}=-\frac{1}{\pi}\frac{\xi^{2}t}{\epsilon+z-2}=-\frac{1}{\pi}\xi^{2}t\int_{1}^{\infty}\left|k\right|^{1-\epsilon-z}dk
\end{equation}
Finally, the term $-\frac{\nu_{0}}{2}\xi^{2}t$ in $U_{0}$ goes to zero. The effective equation for this regime is
\begin{equation}
\frac{\partial \bar{T}}{\partial t} = \frac{1}{2\pi}\int_{1}^{\infty}\left|k\right|^{1-\epsilon-z}dk
\frac{\partial^2 \bar{T}}{\partial y^2}
\end{equation}

\subsubsection{Regime III: $\{4-2z<\epsilon<4,z<2\}\cup\{2<\epsilon<4,z\geq2\}$}

Next, assume that $\partial_{t}^{2}$ is dominant in the expression
 $\partial_{x}^{2z}+\partial_{t}^{2}$.  i.e. $\alpha<z$. 
Observe that $\frac{k^{z}}{k^{2z}+\omega^{2}}
\rightarrow\frac{k^{z}e^{-lz}}{k^{2z}e^{-lz}+\omega^{2}e^{-2\alpha l}}$
converges to $\delta(\omega)$ as $l\rightarrow\infty$, i.e. 
at the fixed point, the Gaussian field $v(k,\omega)=0$ unless $\omega=0$.
Let $\bar{v}(k)=\int v(k,t)dt$. The fixed point has the form
\begin{equation}
e^{-\int\left(\bar{v}(x^{\prime})\partial_{x}^{\epsilon-1}\bar{v}(x^{\prime})\right)dx^{\prime}-i\xi t\bar{v}(x)ds+U_{0}^{\star}}
\end{equation}
and $U_{0}^{\star}$ is thus the same as that of the hyperscaling regime for the steady case
\begin{equation}
U_{0}^{\star}=\xi^{2}t^{2}\int\left|k\right|^{1-\epsilon}\psi_{0}(|k|)dk
\end{equation}
and $\int\left(\bar{v}(x^{\prime})\partial_{x}^{\epsilon-1}\bar{v}(x^{\prime})\right)dx^{\prime}$
and $i\xi t\bar{v}(x)ds$ being both marginal gives
$\alpha=2-\epsilon/2$,
which is the same as that of the hyperscaling regime for the steady case. The effective equation for this regime is
\begin{equation}
\frac{\partial \bar{T}}{\partial t} = t
\left(\frac{1}{\pi}\int_{1}^{\infty}\left|k\right|^{1-\epsilon-z}dk\right)
\frac{\partial^2 \bar{T}}{\partial y^2}
\end{equation}

\subsubsection{Regime IV: $\{4-2z<\epsilon<2,1<z<2\}$}

Solving (\ref{eq:dUdlshort-1}) for $U_{0}$ gives
\begin{equation}
\begin{aligned}
U_{0}(l)=U_{0}(0)-\frac{1}{\pi}\xi^{2}t^{2}\int_{0}^{1}\int_{0}^{1} & \int_{-\infty}^{\infty}\int_{e^{-l}}^{1}e^{i(B_{s}-B_{s^{\prime}})\sqrt{t}k+i\omega(s-s^{\prime})t}\\
& \cdot\left|k\right|^{1-\epsilon}\frac{|k|^{z}}{\omega^{2}+|k|^{2z}}dkd\omega dsds^{\prime}
\end{aligned}
\end{equation}
with $U_{0}(0)=-\frac{\nu_{0}}{2}\xi^{2}t$. The last term under rescaling
$x\rightarrow e^{l}x$, $t\rightarrow e^{\alpha l}t$ becomes 
\begin{equation}
\begin{aligned}
&-\frac{1}{\pi}\xi^{2}t^{2+\frac{\epsilon-2}{z}}e^{((2-\frac{2-\epsilon}{z})\alpha-2)l}\\
&\qquad\cdot\int_{0}^{1}\int_{0}^{1}\int_{-\infty}^{\infty}\int_{e^{-l}}^{1}e^{i(B_{s}-B_{s^{\prime}})e^{\alpha l/2}\sqrt{t}k+i\omega(s-s^{\prime})te^{\alpha l}}\left|k\right|^{1-\epsilon}\frac{|k|^{z}}{\omega^{2}+|k|^{2z}}dkd\omega dsds^{\prime}
\end{aligned}
\end{equation}

We choose $\alpha$ so that $(2-\frac{2-\epsilon}{z})\alpha-2=0$ i.e. 
\begin{equation}
\alpha=\frac{2z}{2z+\epsilon-2}
\end{equation}
Following the same calculations that led to (\ref{eq:G}), we obtain
\begin{equation}
\begin{aligned} & \int_{0}^{1}\int_{0}^{1}\int_{-\infty}^{\infty}e^{i\omega(s-s^{\prime})te^{\alpha l}}\frac{|k|^{z}}{\omega^{2}+|k|^{2z}}d\omega dsds^{\prime}\\
= & \frac{1}{|k|^{z}te^{\alpha l}}-\frac{1}{(|k|^{z}te^{\alpha l})^{2}}(1-e^{-|k|^{z}te^{\alpha l}})
\end{aligned}
\end{equation}
Changing variable $k\rightarrow e^{-\alpha l/z}t^{-1/z}k$, we have
\begin{equation}
V_{0}(l)=-\frac{1}{\pi}\xi^{2}t^{2}e^{2(\alpha-1)l}\int_{e^{(\alpha/z-1)l}t^{1/z}}^{e^{\alpha l/z}t^{1/z}}e^{i(B_{s}-B_{s^{\prime}})e^{\alpha l/2}e^{-\alpha l/z}t^{1/2-1/z}k}\left|k\right|^{1-\epsilon}g(|k|^{z})dk
\end{equation}
where $g(k)=\frac{1}{k}-\frac{1}{k{}^{2}}(1-e^{-k})$. 
Since $z<2$,
the Brownian motion term goes to $0$ as $l\rightarrow\infty$. Therefore
\begin{equation}
V_{0}(l)\rightarrow-\frac{1}{\pi}\xi^{2}t^{2+\frac{\epsilon-2}{z}}\int_{0}^{\infty}\left|k\right|^{1-\epsilon}g(|k|^{z})dk=V_{0}^{\star}
\end{equation}

Notice that $U_{0}(0)=-\frac{\nu_{0}}{2}\xi^{2}t$ vanishes in the limit. The effective equation for this regime is
\begin{equation}
\frac{\partial \bar{T}}{\partial t} = \frac{1}{2\pi}(2+\frac{\epsilon-2}{z})t^{1+\frac{\epsilon-2}{z}}\left(\int_{0}^{\infty}\left|k\right|^{1-\epsilon}g(|k|^{z})dk\right)
\frac{\partial^2 \bar{T}}{\partial y^2}
\end{equation}

\subsubsection{Regime V (nonlocal regime): $\{0<\epsilon<2,z>2\}$}

This regime can be treated in essentially the same way as  the nonlocal regime
for the steady case.
The solution for $U_{0}$ is
\begin{equation}
U_{0}(l)=U_{0}(0)-\frac{1}{\pi}\xi^{2}t^{2}
\int_{0}^{1}\int_{0}^{1}\int_{e^{-l}}^{1}e^{i(B_{s}-B_{s^{\prime}})\sqrt{\nu_0 t}k}\left|k\right|^{1-\epsilon}e^{-|k|^{z}|s-s^{\prime}|t}dkdsds^{\prime}
\end{equation}
with $U_{0}(0)=-\frac{\nu_{0}}{2}\xi^{2}t$. The last term under the rescaling
$x\rightarrow e^{l}x$, $t\rightarrow e^{\alpha l}t$ and the change
of variable $k\rightarrow e^{-\alpha l/2}(\nu_0 t)^{-1/2}k$ becomes 
\begin{equation}
\begin{aligned} 
& -\frac{1}{\pi}e^{(\alpha+\alpha\epsilon/2-2)l}\xi^{2}t^{1+\epsilon/2}\\
&\qquad\cdot\int_{0}^{1}\int_{0}^{1}\int_{e^{(\alpha/2-1)l}\sqrt{t}}^{e^{\alpha l/2}\sqrt{t}}e^{i(B_{s}-B_{s^{\prime}})k}\left|k\right|^{1-\epsilon}e^{-|k|^{z}|s-s^{\prime}|t^{1-z/2}e^{(1-z/2)\alpha l}}dkdsds^{\prime}
\end{aligned}
\end{equation}
Since $z>2$, $1-z/2\rightarrow0$. Choosing $\alpha=\frac{2}{1+\epsilon/2}$,
we obtain a nontrivial fixed point as $l\rightarrow\infty$
\begin{equation}
V_{0}^{\star}=-\frac{1}{\pi}\xi^{2}\nu_0^{\epsilon/2-1}t^{1+\epsilon/2}
\int_{0}^{1}\int_{0}^{1}\int e^{i(B_{s}-B_{s^{\prime}})k}\left|k\right|^{1-\epsilon}dkdsds^{\prime}
\end{equation}
Notice that $U_{0}(0)=-\frac{\nu_{0}}{2}\xi^{2}t$ vanishes in the limit.

Define
\begin{equation}\label{eq:unsteady_alpha}
\nu(\alpha)=Prob\left(
-\frac{1}{\pi}\int_{0}^{1}\int_{0}^{1}\int e^{i(B_{s}-B_{s^{\prime}})k}\left|k\right|^{1-\epsilon}dkdsds^{\prime}\leq\alpha
\right)
\end{equation}
Then the averaged quantity $\bar{T}$ satisfies
\begin{equation}
\bar{T}(x,y,t)=\int\bar{T}_{\alpha}(x,y,t)d\nu(\alpha)
\end{equation}
where
\begin{equation}
\frac{\partial \bar{T}_{\alpha}}{\partial t} = D(\nu_0) \alpha t^{\epsilon/2}  \frac{\partial^2 \bar{T}_{\alpha}}{\partial y^2} 
\end{equation}
for some constant $D(\nu_0)$. 

Note that the effective model for the mean field and nonlocal regimes depends
on the value of $\nu_0$, which is not the case for the other regimes.

\begin{rem}
As discussed in Remark \ref{steady_sum}, we have a space of local stochastic PDEs which is a subspace of the set of all stochastic equations with a noise $v$: $\mathcal{L}\subset\mathcal{E}$, and our initial data of the RG flow is always taken from a one-dimensional space (parametrized by $\nu_0$) $\mathcal{A}\subset\mathcal{L}$. The fixed points we identified above are outside $\mathcal{L}$ for the
nonlocal regime and inside $\mathcal{L}$ for the other regimes. 
Points of $\mathcal{A}$ belong to basins of attraction of different fixed points for the mean field and the nonlocal regime, while for regimes II,III,IV all points of $\mathcal{A}$ are in the same basin of attraction.
\end{rem}

\section{The effective model at the intermediate scales \label{sec:Effective-SPDE}}

In this section we derive the effective stochastic PDEs for $\bar{T}_{l}$ at
arbitrary scale $l$. These equations can be viewed either as analogs of the
models of large-eddy simulation, or models that arise in the Mori-Zwanzig formalism.
We will see that the effective
SPDEs are generally nonlocal, containing kernels in a general form,
in contrast to the local models that are implicitly assumed 
in the Yakhot-Orszag approximate RG scheme \cite{Avellaneda_Approximate_1992}.

Recall that
\begin{equation}
\bar{T}_{l}(x,\xi,t;v_{<})=\langle \hat{T}(x,\xi,t)\rangle_{v_{>}}
\end{equation}

\begin{prop}\label{prop:nonlocal_effPDE}
$\bar{T}_{l}(x,\xi,t;v_{<})$ satisfies the following SPDE:
\begin{equation}
\begin{aligned} & \partial_{t}\bar{T}_{l}+iv_{<}(x,t)\xi\bar{T}_{l}=\frac{\nu_{0}}{2}\partial_{x}^{2}\bar{T}_{l}-\frac{1}{2}\nu_{0}\xi^{2}\bar{T}_{l}+\int_{\mathbb{R}}K_{l}(x,\tilde{x},\xi,t)\hat{T}\big|_{t=0}(\tilde{x},\xi)d\tilde{x}\end{aligned}
\end{equation}
where the kernel $K_{l}$ is a superposition
of kernels over the ensemble of Brownian bridge paths $\{\tilde{B}_{s}:s\in[0,T],\tilde{B}(0)=x,\tilde{B}(t)=\tilde{x}\}$.
\begin{equation}
\begin{aligned}K_{l}(x,\tilde{x},\xi,t)= & \frac{1}{\sqrt{2\pi\nu_{0}t}}e^{-\frac{(x-\tilde{x})^{2}}{2\nu_{0}t}}\mathbb{E}\bigg[e^{-\frac{\nu_{0}}{2}\xi^{2}t}e^{-i\xi\int_{0}^{t}v_{<}(\tilde{B}_{s})ds}\\
 & e^{-\frac{\xi^{2}}{2}\int_{0}^{t}\int_{0}^{t}R(\tilde{B}_{s}-\tilde{B}_{s^{\prime}},s-s^{\prime})dsds^{\prime}}\xi^{2}\int_{0}^{t}R(\tilde{B}_{s}-\tilde{x})ds\bigg]
\end{aligned}
\end{equation}
and
\begin{equation}
R(x,t)=\frac{1}{\pi}\int_{-\infty}^{\infty}\int_{e^{-l}}^{1}e^{ixk+it\omega}\left|k\right|^{1-\epsilon}\frac{|k|^{z}}{\omega^{2}+|k|^{2z}}dkd\omega
\end{equation}
\end{prop}

\begin{proof}
By the Feynman-Kac representation,
\begin{equation}
\hat{T}(x,\xi,t)=\mathbb{E}\bigg[e^{-\frac{\nu_{0}}{2}\xi^{2}t}e^{-i\xi\int_{0}^{t}v(x+\sqrt{\nu_{0}}B_{s},t-s)ds}\hat{T}\big|_{t=0}(x+\sqrt{\nu_{0}}B_{t},\xi)\bigg]
\end{equation}
where $\mathbb{E}$ is the expectation over Brownian motion $B$.
We average out $v_{>}$ and obtain
\begin{equation}
\begin{aligned} & \bar{T}_{l}(x,\xi,t;v_{<})\\
= & \mathbb{E}\bigg[e^{-\frac{\nu_{0}}{2}\xi^{2}t}e^{U_{1}(x,\xi,t)}e^{-\frac{\xi^{2}}{2}\int_{0}^{t}\int_{0}^{t}\left\langle v_{>}(\sqrt{\nu_{0}}(B_{s}-B_{s^{\prime}}),s-s^{\prime})v_{>}(0,0)\right\rangle dsds^{\prime}}\hat{T}\big|_{t=0}(x+\sqrt{\nu_{0}}B_{t},\xi)\bigg]\\
= & \mathbb{E}\bigg[e^{-\frac{\nu_{0}}{2}\xi^{2}t}e^{U_{1}(x,\xi,t)}e^{-\frac{\xi^{2}}{2}\int_{0}^{t}\int_{0}^{t}R_{l}(\sqrt{\nu_{0}}(B_{s}-B_{s^{\prime}}),s-s^{\prime})dsds^{\prime}}\hat{T}\big|_{t=0}(x+\sqrt{\nu_{0}}B_{t},\xi)\bigg]
\end{aligned}
\end{equation}
where $U_{1}(x,\xi,t)=-i\xi\int_{0}^{t}v_{<}(x+\sqrt{\nu_{0}}B_{s},t-s)ds$,
and
\begin{equation}
R_{l}(\sqrt{\nu_{0}}(B_{s}-B_{s^{\prime}}),s-s^{\prime})=\frac{1}{\pi}\int_{-\infty}^{\infty}\int_{e^{-l}}^{1}e^{i\sqrt{\nu_{0}}(B_{s}-B_{s^{\prime}})k+i(s-s^{\prime})\omega}\left|k\right|^{1-\epsilon}\frac{|k|^{z}}{\omega^{2}+|k|^{2z}}dkd\omega
\end{equation}

Now we derive the effective PDE for $\bar{T}_{l}$. Since the generator
of $B_{t}$ is $\partial_{x}^{2}$,
\begin{equation}
\begin{aligned} & \frac{\nu_{0}}{2}\partial_{x}^{2}\bar{T}_{l}=\lim_{r\rightarrow0}\frac{1}{r}\left\{ \mathbb{E}\left[\bar{T}_{l}(x+\sqrt{\nu_{0}}B_{r},\xi,t)-\bar{T}_{l}(x,\xi,t)\right]\right\} \\
= & \lim_{r\rightarrow0}\frac{1}{r}\bigg\{\mathbb{E}\bigg[e^{-\frac{\nu_{0}}{2}\xi^{2}t}e^{-i\xi\int_{0}^{t}v_{<}(x+\sqrt{\nu_{0}}B_{r}+\sqrt{\nu_{0}}B_{s},t-s)ds}\\
 & e^{-\frac{\xi^{2}}{2}\int_{0}^{t}\int_{0}^{t}R(\sqrt{\nu_{0}}(B_{s}-B_{s^{\prime}}),s-s^{\prime})dsds^{\prime}}\hat{T}\big|_{t=0}(x+\sqrt{\nu_{0}}B_{r}+\sqrt{\nu_{0}}B_{t},\xi)\bigg]-[r=0]\bigg\}
\end{aligned}
\end{equation}
where the term $[r=0]$ means the same as the first term except $r=0$.
Because $B_{r}+B_{t}\sim B_{r+t}$ in law,
\begin{equation}
\begin{aligned} & \frac{\nu_{0}}{2}\partial_{x}^{2}\bar{T}_{l}=\lim_{r\rightarrow0}\frac{1}{r}\bigg\{\mathbb{E}\bigg[e^{-\frac{\nu_{0}}{2}\xi^{2}t}e^{-i\xi\int_{r}^{t+r}v_{<}(x+\sqrt{\nu_{0}}B_{s},t-s)ds}\\
 & e^{-\frac{\xi^{2}}{2}\int_{r}^{t+r}\int_{r}^{t+r}R(\sqrt{\nu_{0}}(B_{s}-B_{s^{\prime}}),s-s^{\prime})dsds^{\prime}}\hat{T}\big|_{t=0}(x+\sqrt{\nu_{0}}B_{t+r},\xi)\bigg]-[r=0]\bigg\}\\
= & \lim_{r\rightarrow0}\frac{1}{r}\bigg\{\mathbb{E}\bigg[e^{-\frac{\nu_{0}}{2}\xi^{2}(t+r)}e^{-i\xi\int_{0}^{t+r}v_{<}(x+\sqrt{\nu_{0}}B_{s},t-s)ds}\left(e^{\frac{\nu_{0}}{2}\xi^{2}r+i\xi\int_{0}^{r}v_{<}(x+\sqrt{\nu_{0}}B_{s},t-s)ds}-1+1\right)\\
 & e^{-\frac{\xi^{2}}{2}\int_{0}^{t+r}\int_{0}^{t+r}R(\sqrt{\nu_{0}}(B_{s}-B_{s^{\prime}}),s-s^{\prime})dsds^{\prime}}\left(e^{\frac{\xi^{2}}{2}\int_{\Gamma}R(\sqrt{\nu_{0}}(B_{s}-B_{s^{\prime}}),s-s^{\prime})dsds^{\prime}}-1+1\right)\\
 & \hat{T}\big|_{t=0}(x+\sqrt{\nu_{0}}B_{t+r},\xi)\bigg]-[r=0]\bigg\}
\end{aligned}
\end{equation}
where
\begin{equation}
\Gamma=\{(s,s^{\prime})\in[0,t+r]^{2}\backslash[0,t]^{2}\}
\end{equation}
Therefore
\begin{equation}
\begin{aligned} & \frac{\nu_{0}}{2}\partial_{x}^{2}\bar{T}_{l}=\partial_{t}\bar{T}_{l}+\left(\frac{1}{2}\nu_{0}\xi^{2}+iv_{<}(x,t)\xi\right)\bar{T}_{l}\\
 & +\mathbb{E}\bigg[e^{-\frac{\nu_{0}}{2}\xi^{2}t}e^{-i\xi\int_{0}^{t}v_{<}(x+\sqrt{\nu_{0}}B_{s},t-s)ds}e^{-\frac{\xi^{2}}{2}\int_{0}^{t}\int_{0}^{t}R(\sqrt{\nu_{0}}(B_{s}-B_{s^{\prime}}),s-s^{\prime})dsds^{\prime}}\\
 & \lim_{r\rightarrow0}\frac{1}{r}\left(\frac{\xi^{2}}{2}\int_{\Gamma}R(\sqrt{\nu_{0}}(B_{s}-B_{s^{\prime}}),s-s^{\prime})dsds^{\prime}\right)\hat{T}\big|_{t=0}(x+\sqrt{\nu_{0}}B_{t},\xi)\bigg]
\end{aligned}
\end{equation}
The limit of $\frac{1}{r}\int_{\Gamma}$ as $r\rightarrow0$ with
$\Gamma$ being the infinitesimally thin region defined above is the
line integral times $2$, thus 
\begin{equation}
\begin{aligned} & \frac{\nu_{0}}{2}\partial_{x}^{2}\bar{T}_{l}=\partial_{t}\bar{T}_{l}+\left(\frac{1}{2}\nu_{0}\xi^{2}+iv_{<}(x,t)\xi\right)\bar{T}_{l}\\
 & +\mathbb{E}\bigg[e^{-\frac{\nu_{0}}{2}\xi^{2}t}e^{-i\xi\int_{0}^{t}v_{<}(x+\sqrt{\nu_{0}}B_{s},t-s)ds}e^{-\frac{\xi^{2}}{2}\int_{0}^{t}\int_{0}^{t}R(\sqrt{\nu_{0}}(B_{s}-B_{s^{\prime}}),s-s^{\prime})dsds^{\prime}}\\
 & \xi^{2}\int_{0}^{t}R(\sqrt{\nu_{0}}(B_{s}-B_{t}),s-t)ds\cdot\hat{T}\big|_{t=0}(x+\sqrt{\nu_{0}}B_{t},\xi)\bigg]\\
= & \partial_{t}\bar{T}_{l}+\left(\frac{1}{2}\nu_{0}\xi^{2}+iv_{<}(x)\xi\right)\bar{T}_{l}+\int_{\mathbb{R}}K_{l}(x,\tilde{x},\xi,t)\hat{T}\big|_{t=0}(\tilde{x},\xi)d\tilde{x}
\end{aligned}
\end{equation}
where 
\begin{equation}
\begin{aligned}K_{l}(x,\tilde{x},\xi,t)= & \frac{1}{\sqrt{2\pi\nu_{0}t}}e^{-\frac{(x-\tilde{x})^{2}}{2\nu_{0}t}}\mathbb{E}\bigg[e^{-\frac{\nu_{0}}{2}\xi^{2}t}e^{-i\xi\int_{0}^{t}v_{<}(\tilde{B}_{s},t-s)ds}\\
 & e^{-\frac{\xi^{2}}{2}\int_{0}^{t}\int_{0}^{t}R(\tilde{B}_{s}-\tilde{B}_{s^{\prime}},s-s^{\prime})dsds^{\prime}}\xi^{2}\int_{0}^{t}R(\tilde{B}_{s}-\tilde{x})ds\bigg]
\end{aligned}
\end{equation}
and $\tilde{B}$ is a Brownian bridge on $[0,t]$ with variance $\nu_{0}$
and $\tilde{B}(0)=x$, $\tilde{B}(t)=\tilde{x}$. 
\end{proof}

{\bf Acknowledgement}.  The work presented here is supported in part by the
DOE grant DE-SC0009248 and the ONR grant N00014-13-1-0338.

\bibliographystyle{plain}
\bibliography{shear}

\end{document}